\documentclass[twoside,11pt]{article}
\usepackage{arxiv} 
\usepackage{hyperref}
\usepackage{changes}
\usepackage{url}
\usepackage{amssymb}
\usepackage{enumitem}
\usepackage{graphicx}
\usepackage{color,xcolor}
\usepackage{subcaption}
\usepackage{hyperref}
\usepackage{amsmath}
\usepackage{amssymb}
\usepackage{amsthm}
\usepackage{graphicx}
\usepackage{tabularx}
\usepackage{fancyvrb}
\usepackage{comment}
\usepackage{algorithm}
\usepackage{algpseudocode}
\usepackage{float}
\usepackage{wrapfig}
\usepackage{listings}
\usepackage{natbib}
\usepackage{mathtools}
\usepackage{booktabs}
\usepackage{siunitx}
\usepackage{thm-restate}

\newcommand{\argmax}{\mathop{\rm argmax}}

\newtheorem{lem}{Lemma}[section]

\newtheorem{asmp}{Assumption}[section]
\newtheorem{defn}{Definition}[section]

\newtheorem{fact}{Fact}[section]

\newtheorem{rem}{Remark}[section]

\def\approxcorrect{\checkmark\kern-1.1ex\raisebox{.89ex}{$\times$}}

\usepackage{amsmath,amsfonts,bm}

\def\eqref#1{equation~\ref{#1}}

\def\1{\bm{1}}

\DeclareMathAlphabet{\mathsfit}{\encodingdefault}{\sfdefault}{m}{sl}
\SetMathAlphabet{\mathsfit}{bold}{\encodingdefault}{\sfdefault}{bx}{n}

\def\gA{{\mathcal{A}}}

\def\gF{{\mathcal{F}}}
\def\gG{{\mathcal{G}}}

\definecolor{dkgreen}{rgb}{0,0.6,0}
\definecolor{gray}{rgb}{0.5,0.5,0.5}
\definecolor{mauve}{rgb}{0.58,0,0.82}

\allowdisplaybreaks

\ShortHeadings{}{}
\firstpageno{1}
\begin{document}
\title{Taming the Exponential Action Set: Sublinear Regret and Fast Convergence to Nash Equilibrium in Online Congestion Games}
\author{\name Jing Dong \email jingdong@link.cuhk.edu.cn  \\
       \addr The Chinese University of Hong Kong, Shenzhen\\
       \name Jingyu Wu \email wujingyu@mail.ustc.edu.cn \\ \addr University of Science and Technology of China\\
       \name Siwei Wang \email siweiwang@microsoft.com \\
       \addr Microsoft Research Asia \\
       \name Baoxiang Wang \email bxiangwang@cuhk.edu.cn  \\
       \addr The Chinese University of Hong Kong, Shenzhen\\
       \name Wei Chen \email weic@microsoft.com \\
       \addr Microsoft Research Asia}
\maketitle

\begin{abstract}
The congestion game is a powerful model that encompasses a range of engineering systems such as traffic networks and resource allocation. It describes the behavior of a group of agents who share a common set of $F$ facilities and take actions as subsets with $k$ facilities. In this work, we study the online formulation of congestion games, where agents participate in the game repeatedly and observe feedback with randomness. We propose CongestEXP, a decentralized algorithm that applies the classic exponential weights method. By maintaining weights on the facility level, the regret bound of CongestEXP avoids the exponential dependence on the size of possible facility sets, i.e., $\binom{F}{k} \approx F^k$, and scales only linearly with $F$. Specifically, we show that CongestEXP attains a regret upper bound of $O(kF\sqrt{T})$ for every individual player, where $T$ is the time horizon. On the other hand, exploiting the exponential growth of weights enables CongestEXP to achieve a fast convergence rate. If a strict Nash equilibrium exists, we show that CongestEXP can converge to the strict Nash policy almost exponentially fast in $O(F\exp(-t^{1-\alpha}))$, where $t$ is the number of iterations and $\alpha \in (1/2, 1)$.  
%

\end{abstract}

\section{Introduction}


Congestion games are a class of general-sum games that can be used to describe the behavior of agents who share a common set of facilities (resources) \citep{brown1949some}. 
In these games, each player chooses a combination of facilities, and popular facilities will become congested, yielding a lower utility for the players who choose them. Thus, players are incentivized to avoid congestion by choosing combinations that are less popular among the other players. A range of real-world scenarios can be captured by the congestion game model, such as traffic flow, data routing, and wireless communication networks \citep{tekin2012atomic,cheung2014congestion,zhang2020stochastic}. 

In the model of the congestion game, the Nash equilibrium is a popular concept to describe the behavior of selfish players and the dynamics induced by decentralized algorithms. It describes a stable state of the game where no player can improve their utility by unilaterally changing their choice of actions. When a unique Nash equilibrium exists in the congestion game, it can be a reference point for players to coordinate to avoid suboptimal outcomes. 
Beyond the Nash equilibrium, social welfare is a significant metric, capturing the overall utility or well-being of all players involved. It serves as a crucial benchmark, enabling the evaluation of the efficiency loss incurred when transitioning from centralized to decentralized algorithms.

In the classic one-shot congestion game setting, the Nash equilibrium and the loss of efficiency due to decentralized dynamics have been well studied \citep{roughgarden2002bad,roughgarden2007routing}. However, these results do not provide insights into how players arrive at the equilibrium. This motivates the study of congestion games in an online learning framework, where players participate in the game repeatedly at every time step. This framework better models many realistic scenarios, such as the traffic congestion problem in urban areas. In this repeated congestion game setting, players such as drivers in a congested city must choose between different routes to reach their destinations every day. As more drivers use a particular route, the congestion on that route increases, leading to higher travel times and lower utility. In this scenario, players can update their desired route every day to optimize their utility, but the observed utility by each player may be subject to randomness due to uncertainty in the actual congestion situation (e.g., the influence of the weather). All these make it suitable to model the congestion game in an online learning framework. 

While there have been various decentralized algorithms that can attain the Nash equilibria efficiently for the general online games, they can suffer from a linear dependency on the number of actions when directly applied to the congestion game \citep{daskalakis2011near,syrgkanis2015fast,chen2020hedging,hsieh2021adaptive,daskalakis2021near,giannou2021rate}, which is exponential with $k,F$. On the other hand, algorithms designed specifically for congestion games either only converge to Nash equilibria asymptotically \citep{kleinberg2009multiplicative,krichene2015online,palaiopanos2017multiplicative}, on average \citep{cuilearning22}, or require additional assumptions on the structure of the game \citep{chen2015playing,chen2016generalized}. Moreover, to the best of our knowledge, there is no algorithm that can simultaneously guarantee both low regret and fast convergence to Nash equilibrium for each player. While some online learning algorithms, such as exponential weights, have been shown to converge faster than others due to the specific choice of regularization \citep{giannou2021rate}, previous regret results indicate that their guarantees still rely on the exponentially large number of actions, due to their specific form of updates (exponential weighting) \citep{cohen2016exponentially}.

In this paper, we study the online congestion game with semi-bandit and full information feedback. We propose a decentralized algorithm that modifies the celebrated exponential weights algorithm, which can be utilized by each player without additional information about other players' utility. 
From the individual player's perspective, we show that the algorithm guarantees sublinear individual regret, with respect to the best action in hindsight when holding the other player's strategy fixed. 
We remark that the regret is also only linear with respect to the number of facilities. As a result of this, we show that the optimal social welfare can be efficiently approximated, up to an error that is only linear with respect to the number of facilities. When a strict Nash equilibrium exists for the congestion game, we also prove that our algorithm is capable of converging to the strict Nash equilibrium fast, with an almost exponentially fast rate that is only linear with respect to the number of facilities. 
\section{Related works}

\paragraph{Learning in games}
Online learning has a long history that is closely tied to the development of game theory. The earliest literature can be traced back to Brown's proposal on using fictitious play to solve two-player zero-sum games \citep{brown1949some}. 
It is now understood that fictitious play can converge very slowly to Nash equilibrium \citep{daskalakis2014counter}.
On the other side, it has been shown that if each player of a general-sum, multi-player game experiences regret that is at most $f(T)$, the empirical distribution of the joint policy converges to a coarse correlated equilibrium of the game with a rate of $O(f(T)/T)$ \citep{cesa2006prediction}. This implies that a variety of online learning algorithms such as Hedge and Follow-The-Regularized-Leader algorithms can converge to coarse correlated equilibria at a rate of $O(1/\sqrt{T})$.

While the standard no-regret learning dynamic can guarantee convergence to equilibria, it has been shown that more specialized no-regret learning protocols can do better \citep{daskalakis2011near,syrgkanis2015fast,chen2020hedging,hsieh2021adaptive,daskalakis2021near}. 
It has also been shown that when strict pure Nash equilibria are present, algorithms that are based on entropic regularization (e.g. exponential weights) can converge fast to the equilibria \citep{cohen2016exponentially,giannou2021rate}. Moreover, such convergence rate holds for a variety of different feedback models, from full information to bandit feedback. 

Though all of the above-mentioned methods are applicable to congestion games, the results usually involve a linear dependency on the number of actions. As each action is a combination of the different facilities (resources) in the congestion games, the results lead to the undesirable  exponential dependency on the number of facilities. 

\paragraph{Learning in online congestion games}
Congestion games were first introduced in the seminal work \citet{rosenthal1973class} as a class of games with pure-strategy Nash equilibria. It has then been extensively studied, where its Nash equilibria have been characterized in \citet{roughgarden2002bad} and a comprehensive introduction has been given in \citet{roughgarden2007routing}. 

In the online setting, many works use no-regret learning to develop learning dynamics in this class of games for efficient convergence. \citet{kleinberg2009multiplicative} are the first to study no-regret learning for congestion games. They showed that the well-known multiplicative weights learning algorithm results in convergence to pure equilibria. Furthermore, they identified a set of mixed Nash equilibria that are weakly stable and showed that the distribution of play converges to this set. 
Followup works \citet{krichene2015online} showed that multiplicative weights algorithms converge to the set of Nash equilibria in the sense of Ces\`{a}ro means, and \citet{palaiopanos2017multiplicative} investigated the effect of learning rate on convergence to Nash equilibria. 

With an additional assumption of convex potential functions, \citet{chen2015playing,chen2016generalized} established a non-asymptotic convergence rate. However, their rate has an exponential dependency on the number of facilities. \citet{cuilearning22} gave the first non-asymptotic convergence rate under semi-bandit feedback and without an exponential dependency on the number of facilities. However, the convergence is with respect to the averaged-over-time policy and suffers from a dependency of $F^9$, where $F$ is the number of facilities. This result is later improved by concurrent work \citet{panageas23}, who proposed an online stochastic gradient descent algorithm that convergences to an $\epsilon$-approximate Nash equilibrium in $O(\epsilon^{-5})$ time while each individual player enjoys a regret of $O(T^{4/5})$. 

\paragraph{Combinatorial bandits and shortest path}
Combinatorial bandits offer an extension of the classic multi-armed bandit problem where the player must select an action that involves a combination of various resources \citep{cesa2012combinatorial,chen2013combinatorial,lattimore2020bandit}. In a special case, the shortest path problem can be viewed as a combinatorial bandit problem where the resources are edges on a graph and the action is a path \citep{gyorgy2007line}. Efficient algorithms have been proposed for these problems, and it has been shown that the sublinear regret only linearly depends on the number of resources. 
However, it is important to note these algorithms are designed for a single player, and as a result, they may not converge to a Nash equilibrium when applied directly to congestion games by allowing each player to execute the algorithm.
\section{Problem Formulation}

\subsection{Congestion game} 
A congestion game with $n$ players is defined by $\gG = \left(\gF, \{\gA_i\}_{i=1}^n, \{r^f\}_{f \in \gF}\right)$, where
i) $\gF$ is the facility set that contains $F$ facilities;
ii) $\gA_i$ is the action space for player $i$ and contains $A$ actions (we assume that the action space for each player is the same), where each action $a_i \in \gA_i$ is a combination of $k$ facilities in $\gF$;
and iii) $r^f: (\gA_1 \times \dots \times \gA_n) \rightarrow [0,1]$ is the reward function for facility $f\in \gF$, which only depends on the number of players choosing this facility, i.e., $\sum^n_{i=1} \mathbb{I} \{f \in a_i\}$.
We denote $a = (a_i, a_{-i})$ as a joint action, where $a_{-i}$ is the actions of all other players except player $i$.   
The total reward collected by player $i$ with joint action $a = (a_i, a_{-i})$ is $r_i(a_i, a_{-i}) = \sum_{f \in a_i} r^f(a_i, a_{-i})$. 
Without loss of generality, we assume that $r^f(a) \in [0,1]$.

Deterministically playing actions $a = (a_i, a_{-i})$ is referred to as a pure strategy. 
The player can also play a mixture of pure strategy, $\omega_i \in \Delta(\gA_i)$, where $\Delta(\gA_i)$ denotes the probability simplex of action space $\gA_i$. Similarly, we use $\omega = (\omega_i, \omega_{-i})$ to denote a joint randomized policy. 

\subsection{Nash equilibrium}
One of the commonly used solution concepts in congestion games is Nash equilibrium (NE), and the policies that lead to a Nash equilibrium are referred to as Nash policies. The players are said to be in a Nash equilibrium when no player has an incentive from deviating from its current policy (as described in the definition below). 
\begin{defn}[Nash equilibrium]\label{def:nash}
A policy $\omega^\ast = (\omega^\ast_1, \ldots, \omega^\ast_n)$ is called a \textbf{Nash equilibirum} if 
for all $i \in [n]$, $r_i(\omega_i^\ast, \omega_{-i}^\ast) \geq  r_i(\omega_i, \omega_{-i}^\ast) \,, \forall \omega_i \in \Delta(\gA_i)$.
When $\omega^\ast$ is pure, the equilibrium is called a \textbf{pure Nash equilibrium}. 
In addition, when the strategy is pure and the inequality is a strict inequality, 
the equilibrium is called a \textbf{strict Nash equilibrium}. 
\end{defn}

\begin{fact}[\citep{rosenthal1973class}]
There exists a pure Nash equilibrium in any congestion game.
\end{fact}



\subsection{Social welfare and price of anarchy}

Except for Nash equilibrium, another commonly used metric to measure the efficiency of the dynamics between the players is through social welfare. 
For a given joint action $a = \{a_i\}_{i=1}^n$, the social welfare is defined to be the sum of players' rewards, i.e., $W(a) = \sum^n_{i=1} r_i(a)$, and the optimal social welfare of the game is defined as 
\begin{align*}
    \mathrm{OPT} = \max_{a \in \mathcal{A}_1 \times \dots \times \mathcal{A}_n} W(a) \,.
\end{align*}
This optimality is under the case where a central coordinator could dictate each player's strategy, and each player's individual incentives are not considered. 

Based on the definition of $\mathrm{OPT}$, We can define smooth games as follows. 
\begin{defn}[Smooth game \citep{roughgarden2009intrinsic}]\label{def:smooth}
A game is $(\lambda, \mu)$-smooth if there exists a joint action $a^\ast$ such that for any joint action $a$, 
$\sum_{i \in n} r_i\left(a_i^\ast, a_{-i}\right) \geq \lambda \mathrm{OPT}-\mu W(a) 
$.
\end{defn}

\citet{nisan2007algorithmic} show that congestion games are smooth when the reward function are affine, that is, when 
	$r^f(a)$ is an affine function on the scalar variable $\sum^n_{i=1} \mathbb{I} \{f \in a_i\}$. 
%
This property enables certain decentralized no-regret learning dynamics to efficiently approximate the optimal welfare \citep{syrgkanis2015fast}. 



\subsection{Online congestion game} 

In this paper, we study the congestion game in an online setting with a finite time horizon $T$, where the underlying reward function is unknown. At each time step $t \in [T]$, each player chooses (randomized) policy $\omega_i^t$, from which it forms a joint policy $\omega^t = (\omega_1^t, \ldots, \omega_n^t)$. Then each player $i$ draws a random action $a_i^t \sim \omega_i^t$, plays this action (denote $a^t$ the joint action), and receives overall reward of $\sum_{f\in a_i^t} R^f(a^t)$, where $R^f(a^t)$'s are random variables that satisfy the following assumption. 
%
\begin{asmp}\label{asmp:sto_rewards}
For any facility $f \in \mathcal{F}$, any joint action $a^t$ and any player $i \in [n]$, let $\mathcal{H}_t$ be the history up to time step $t-1$. 
Then, $R^f(a) \in [0,1]$, and $\mathbb{E}\left[R^f(a^t)\mid \mathcal{H}_t\right] = r^f(a^t)$.
\end{asmp}

The assumption implies that the mean of $R^f(a^t)$ is always $r^f(a^t)$. Hence the Nash equilibrium and expected social welfare of the online congestion game is the same as those of the offline congestion game.

We consider two types of feedback rules in this paper: \textit{semi-bandit feedback}, and \textit{full information feedback}.
In the semi-bandit feedback, player $i$ observes all the $R^f(a^t)$'s for any $f \in a_i$ (only the facilities he played); and in full information feedback, player $i$ observes all possible information $R^f(a_i,a_{-i})$, for every $a \in \mathcal{A}_i$, $\forall f \in a$. 
The efficiency of a sequence of policy $\{\omega_i^t\}_{t=1}^T$ can be measured by the individual regret of all the players (which is defined as follows). 
\begin{defn}[Individual regret]\label{def:individual_regret}
The individual regret of player $i$ playing policy $\{\omega_i^t\}_{t=1}^T$ is defined as the cumulative difference between the received rewards and the rewards incurred by a best-in-hindsight policy, that is 
\begin{align*}
    \mathrm{Regret}_i(T) = \max_{\omega_i \in \gA_i, \{\omega_{-i}^t\}^T_{t=1}} \sum^T_{t=1} r_i(\omega_i, \omega_{-i}^t) - r_i(\omega_i^t, \omega_{-i}^t) \,.
\end{align*}
\end{defn}



\section{Algorithm}

\begin{algorithm}[t]
\caption{CongestEXP}\label{alg}
\begin{algorithmic}[1]
\State \textbf{Input: }learning rate $\eta$. 
\State For all $f \in \mathcal{F}$, initialize $\tilde{y}_i^0(f) = 0$, $\omega_i^0(a)$ initialized according to Equation (\ref{eq:alg7}).
\For{$t = 1, \ldots, T$}
\State Players play strategy $a^t \sim \omega^t$, $\omega^t = (\omega_1^t, \ldots, \omega_n^t)$. 
\State Each player $i$ observe $R^f(a^t) \sim r^f(a^t)$, for each $f \in a_i^t$.
\State 
Each player $i \in [N]$ computes $\Tilde{y}_i^{t}(f)$ according to Equation (\ref{eq:alg6}).
\State Each player $i \in [N]$ updates $\omega_i^{t+1}(a)$ according to Equation (\ref{eq:alg7}).
\EndFor
\end{algorithmic}
\end{algorithm}

In this section, we introduce CongestEXP, a decentralized algorithm for online congestion games (Algorithm \ref {alg}). 

The algorithm uses the combinatorial nature of the action space. 
Each player maintains a sampling distribution $\omega_i^t$ and a facility-level reward estimator $\tilde {y}_i^t$.
At each time step, they first draw a random action $a_i^t \sim \omega_i^t$ and play this action.
Then they use their received information to update $\tilde {y}_i^t(f)$'s (for all $f\in \gF$) as follows
\begin{align}\label{eq:alg6}
    \Tilde{y}_i^{t}(f) =  1 - \frac{\mathbb{I} \{f \in a_i^t\}(1 - R^f(a^t))}{q_i^t(f)} \,, \quad q_i^t(f) = \sum_{a_i \in \gA_i, f \in a_i}\omega_i^t(a_i) \,,
\end{align}
where $q_i^t(f)$ is  the probability that player $i$ selects facility $f$ at time $t$ based on its current policy $\omega_i^t$. 
%
One can easily check that if $f \in a_i^t$, $\Tilde{y}_i^{t}(f)$ is an unbiased estimator for $r^f(a^t)$, and with these facility-level reward estimators, the players then update $\omega_i^{t+1}$ as follows (exponential weighting), and then proceed to the next time step.
\begin{align}\label{eq:alg7}
    \omega_i^{t+1}(a) = \frac{\prod_{f\in a}\Tilde{\omega}_i^{t}(f)}{\sum_{a_i \in \gA_i}\prod_{f^\prime \in a_i}\Tilde{\omega}_i^{t}(f^\prime)} \,, \forall a \in \gA_i\,, \quad \text{where}\  \ \Tilde{\omega}_i^{t}(f) = \exp\left(\eta \sum^t_{j=1}\Tilde{y}_i^{j}(f)\right) \,.
\end{align}

On the one hand, in the semi-bandit setting, our algorithm leverages this kind of feedback and estimates rewards at the facility level.
We note that this idea has also been previously utilized to tackle online shortest path problems and combinatorial bandit problems, as documented in literature \citep{gyorgy2007line,cesa2012combinatorial,chen2013combinatorial,combes2015combinatorial}.
This enables us to achieve a low individual regret (Theorem \ref{thm:regret}) and guarantee a lower bound for the overall social welfare (Corollary \ref{cor:poa}). 
On the other hand, our algorithm constructs exponential weights based on the reward estimation at the action level. 
This makes sure that the joint policy $\omega^t$ can converge to a Nash equilibrium fast  when it is nearby (Theorem \ref{thm:nash} and \ref{thm:stoc_nash}).

In summary, our results indicate that adopting Algorithm \ref{alg} in a congestion game leads to favorable outcomes. Each player enjoys favorable cumulative individual rewards, without compromising the overall social welfare. Moreover, when the joint policy is close to the Nash equilibrium, players can quickly converge to a stable equilibrium state, avoiding inefficient and chaotic dynamics. 
\section{Theoretical Results}

In this section, we present our main theoretical results.


\subsection{Sublinear individual regret with linear dependency on $F$}
Our first theorem shows that each individual player enjoys a sublinear individual regret.

\begin{restatable}[]{thm}{regret}\label{thm:regret}
    Under semi-bandit feedback, Algorithm \ref{alg} with $\eta = \frac{1}{\sqrt{T}}$ satisfies that for all $i\in [n]$,  
    \begin{equation*}
     \text{Regret}_i(T)  =  O \left(k F \sqrt{T} \right).
    \end{equation*}
\end{restatable}

Compared with naively applying exponential weights on the congestion game (with a regret of $\Tilde{O}\left(\sqrt{A_i T}\right)$ \citep{auer2002nonstochastic}), we can see that Theorem \ref{thm:regret} reduces the factor $\sqrt{A_i}$ to $kF$.
This is a significant improvement since $A \approx F^k$ is exponentially larger than $kF$.

Though there exist some works to achieve a similar regret upper bound \citep{daskalakis2021near}, we emphasize that these algorithms only work in the full-information setting, but not the semi-bandit setting.
Besides, our algorithm can converge to a strict Nash equilibrium fast, while existing ones can only guarantee to converge to a coarse correlated equilibrium (please see details in Section 5.2).

\paragraph{Tight approximation to optimal welfare}
One immediate consequence of Theorem \ref{thm:regret} is that our proposed algorithm can achieve a tight approximation to the optimal social welfare. 
\begin{restatable}[]{cor}{poa}\label{cor:poa}
Under semi-bandit feedback, if the congestion game is $(\lambda, \mu)$-smooth, then Algorithm \ref{alg} with $\eta = \frac{1}{\sqrt{T}}$ satisfies 
    \begin{align*}
        \frac{1}{T} \sum^T_{t=1} W(\omega^t) \geq \frac{\lambda}{1 + \mu} \mathrm{OPT} - O \left(\frac{nk F }{\sqrt{T}(1+\mu)} \right)\,.
    \end{align*}
\end{restatable}
We remark that $\frac{\lambda}{1 + \mu} \mathrm{OPT}$ is shown to be a tight approximation of optimal social welfare possible by offline algorithms that attain Nash equilibrium in congestion games \citep{roughgarden2009intrinsic}. Therefore, the above result shows that our algorithm is as efficient as any offline Nash policy asymptotically. 

\paragraph{Technical highlight of Theorem \ref{thm:regret}} 

In classical proofs of exponential weights algorithms, the regret is closely linked to the quadratic term of the reward estimator, i.e., $\mathbb{E}_t[\sum_{a_i \in \mathcal{A}_i} \omega_i^t(a_i) \left( \tilde{y}_i^t(a_i)\right)^2 ]$\citep{auer2002nonstochastic,lattimore2020bandit}, where $\tilde{y}_i^t(a_i)$ is the estimated reward of action $a_i$ at time step $t$, and $\mathbb{E}_t[\cdot]$ denotes the conditional expectation over all history up to time $t$. If we can upper bound this term by a constant polynomial with $k$ and $F$, then we can remove the exponential factor in the individual regret upper bound.  

With our facility level estimator (Eq. (\ref{eq:alg6})), $\tilde{y}_i^t(a_i) = \sum_{f \in a_i} \tilde{y}_i^t(f)$, the above term could be upper bounded as:  
\begin{align}
\nonumber   &\sum_{a_i \in \mathcal{A}_i} \omega_i^t(a_i) \bigg(\sum_{f \in a_i} \tilde{y}_i^t(f)\bigg)^2 \\
\label{Eq_999}    \leq \ & k \sum_{a_i \in \mathcal{A}_i} \omega_i^t(a_i) \sum_{f \in a_i} \bigg(1 - \frac{\mathbb{I} \{f \in a_i^t\}(1 - R^f(a^t_i, a^t_{-i}))}{q_i^t(f)}\bigg)^2 \\
\label{Eq_998} = \ & k+k\sum_{a_i \in \mathcal{A}_i}  \omega_i^t(a_i) \sum_{f \in a_i}\left(\bigg(\frac{\mathbb{I}\left\{f \in a_i^t\right\}(1 - R^f(a^t_i, a^t_{-i}))}{q_i^t(f)}\bigg)^2-\frac{2 \mathbb{I}\left\{f \in a_i^t\right\}(1 - R^f(a^t_i, a^t_{-i}))}{q_i^t(f)}\right) \\
\nonumber \leq \ & k+k\sum_{f \in \mathcal{F}}\bigg(\frac{\mathbb{I}\left\{f \in a_i^t\right\}(1 - R^f(a^t_i, a^t_{-i}))}{q_i^t(f)}\bigg)^2 \sum_{a_i \in \mathcal{A}, f \in a_i} \omega_i^t(a_i)\\
\label{Eq_997}\leq \ & k+k\sum_{f \in \mathcal{F}}\bigg(\frac{\mathbb{I}\left\{f \in a_i^t\right\}(1 - R^f(a^t_i, a^t_{-i}))}{q_i^t(f)}\bigg)^2 q_i^t(f)\,,
\end{align}
where Eq. (\ref{Eq_999}) is by the Cauchy-Shwarz inequality, Eq. (\ref{Eq_998}) is by noting $\sum_{a_i \in \mathcal{A}_i} \omega_i^t(a_i) = 1$, and Eq. (\ref{Eq_997}) is by the definition of $q_i^t(f)$.
Notice that $(1 - R_i^f(a^t))^2$ is upper bounded by $1$ and taking a conditional expectation over all history up to time $t$ yields (denote $\mathbb{E}_t$ as the conditional expectation operator, and the $q_i^t(f)$ cancels out the $\mathbb{I}\{f \in a_i^t\}^2$). 
\begin{align}\label{eq:quadratic}
    \mathbb{E}_t\Bigg[\sum_{a_i \in \mathcal{A}_i} \omega_i^t(a_i) \bigg(\sum_{f \in a_i} \tilde{y}_i^t(f)\bigg)^2 \Bigg] 
    \leq \  k + kF \,,
\end{align}
and this is an upper bound polynomial with $k$ and $F$.

From the above explanation, one can see the necessity of estimating the rewards at the facility level. 
If the reward estimator is constructed at the action level, that is, an estimator of the form 
\begin{equation*}
    \tilde{y}_i^t(a_i) = k - \frac{\mathbb{I}\left\{a_i = a_i^t\right\}\left(k- \sum_{f \in a_i^t} R_i^f\left(a_i^t, a_{-i}^t\right)\right)}{\omega_i^t(a_i)},
\end{equation*}
Consider the case that $R_i^f\left(a_i^t, a_{-i}^t\right)$ is always 0 and at the beginning $\omega_i^t(a_i) = 1/|\mathcal{A}_i|$, then this quadratic term $\mathbb{E}_t[\sum_{a_i \in \mathcal{A}_i} \omega_i^t(a_i) \left( \tilde{y}_i^t(a)\right)^2 ]$ is approximately
\begin{align*}
     \mathbb{E}_t\Bigg[\sum_{a_i \in \mathcal{A}_i} \omega_i^t(a_i) \left( \tilde{y}_i^t(a)\right)^2 \Bigg] \approx \sum_{a_i \in \mathcal{A}_i} \omega_i^t(a_i) \cdot \left(\omega_i^t(a_i) \left({k\over \omega_i^t(a_i)}\right)^2\right) = k^2|\mathcal{A}_i| \,,
\end{align*}
which scales with the number of actions and is thus always exponentially large. 

\subsection{Fast convergence to strict Nash equilibrium}
Beyond the low individual regret, we also show that our algorithm can produce a set of policies $\{\omega_i^t\}^T_{t=1}$ that converges fast to a strict Nash equilibrium $\omega_i^\ast$ in the full-information setting.




We first consider a simple case, where each player observes the expected rewards directly (which also take expectation on the randomness of $a_{-i}^t \sim \omega_{-i}^t$, i.e. $\mathbb{E}_{a^t_{-i} \sim \omega_{-i}^t} [r^f(a_i, a^t_{-i})]$ for any $a_i$). In addition, we assume that any $k$ facilities form an action. 
\begin{restatable}[]{thm}{thmnash}\label{thm:nash}
Consider the case where each player receives $\mathbb{E}_{a^t_{-i} \sim \omega_{-i}^t} [r^f(a^t_i, a^t_{-i})], \forall a_i \in \mathcal{A}_i, \forall f \in a_i$ in a game that permits a strict Nash equilibrium $\omega^\ast = (\omega_1^\ast, \cdots, \omega_n^\ast)$, and let $\Tilde{y}_i^{t}(f) =  \mathbb{E}_{a^t_{-i} \sim \omega_{-i}^t} [r^f(a^t_i, a^t_{-i})]$ in Line 6 of Algorithm \ref{alg}. Suppose $\Tilde{y}_i^0(f), \forall i \in [n]$ is initialized such that $\omega^0 \in U_{M} \subseteq U_\epsilon$, then for any $i \in [n]$ and any $t$, we have \[\left\| \omega_i^t - \omega_i^\ast\right\|_1
        \leq  2(kF \exp(- M - \eta \epsilon t)) \,,\]
    where $M \geq \left|\log\left(\frac{\epsilon}{2kF}\right)\right| $, and $\epsilon$ is a constant that is game-dependent only. 
\end{restatable}
\begin{rem}
    We note that the convergence rate of the algorithm can be improved by increasing the step size $\eta$. This is because when each player receives expected rewards, the player can take greedy steps toward the equilibrium strategy. 
    This agrees with greedy strategies that are previously employed to reach strict Nash equilibrium \citep{cohen2016exponentially}.
    However, such greedy policies would not work in the presence of reward uncertainty, as we will discuss in Theorem \ref{thm:stoc_nash}.
\end{rem}

It is worth mentioning that the convergence rate of our algorithm does not rely on the number of actions $A_i$, but rather solely on the number of facilities $F$. This is an improvement over the previous findings for exponential weights algorithms with non-combinatorial action spaces in the context of a congestion game, where the rate is linearly dependent on the number of actions \citep{cohen2016exponentially}. The reason for this is the utilization of our facility-level reward estimation technique once again.

Previous studies on the convergence rate of congestion games \citep{chen2015playing,chen2016generalized} have established a linear rate of convergence when the game possesses a smooth potential function and the algorithm is given an appropriate initial starting point. 
The potential function provides a means to capture the incentives of all players to modify their actions and can be used to characterize the dynamics in policy updates. 
Assuming the smoothness of the potential function implies optimization on a simpler policy optimization landscape.
In contrast, our algorithm achieves a much faster rate of convergence, and this convergence rate still holds even in the absence of a smooth potential function. 
We adopt a different approach where we directly argue through the algorithm update rule that the updated policy will always fall within a neighborhood around the Nash equilibrium. 
This bypasses the need for a smoothness potential function and demonstrates the effectiveness of our approach.


In addition, we remark that though some variants of Mirror Descent (MD) or Follow-the-Regularized-Leader (FTRL) algorithms are also proven to enjoy sublinear regret with logarithmic dependency on action space in the full information setting \citep{daskalakis2021near},
these results only imply convergence to an approximate coarse correlated equilibrium and do not enjoy convergence to Nash equilibrium. 
In comparison, Nash equilibrium is much more stable, as the dynamic will remain there unless external factors change, while coarse correlated equilibrium may be more sensitive to small changes in the correlation method, which can lead to deviation from the equilibrium \citep{nisan2007algorithmic}.


To prove Theorem \ref{thm:nash}, we first identify that there exists a neighborhood around the strict Nash equilibrium, such that for any player $i$, his action in the strict Nash equilibrium is the only optimal choice. 
\begin{lem}\label{lem:neighborhood}
    If there exists a strict Nash equilibrium $a^\ast = (a_1^\ast, \ldots, a_n^\ast)$, then there exists $\epsilon > 0$ and a neighborhood $U_\epsilon$ of $a^\ast$, such that for all $\Tilde{\omega} = (\Tilde{\omega}_{i}, \Tilde{\omega}_{-i}) \in U_\epsilon$, \[ r_i(a_i^\ast, \Tilde{\omega}_{-i}) - r_i(a_i, \Tilde{\omega}_{-i}) \geq \epsilon \,, \quad \forall i \in [n]\,, a_i \in \gA_i \,, a_i \neq a_i^\ast \,,\]
    where $ r_i(a_i, \omega_{-i})$ is defined as $\mathbb{E}_{a_{-i} \sim \omega_{-i}}[ r_i(a_i, a_{-i})]$.
\end{lem}
\begin{proof}
    By the existence of the strict Nash equilibrium, there exists $\epsilon > 0$ such that $\forall a_i \in \gA_i$, $r_i\left(a_i, a_{-i}^\ast\right) \le r_i(a_i^\ast, a_{-i}^\ast) - 2\epsilon$. Then by continuity, we know that for $(\Tilde{\omega}_{i}, \Tilde{\omega}_{-i}) \in U_\epsilon$, $r_i(a_i,\Tilde{\omega}_{-i}) \leq r_i(a_i^\ast, \Tilde{\omega}_{-i}) - \epsilon$. 
\end{proof}
Moreover, if the difference in reward estimator $\Tilde{z}_i^t(a_i) = \sum^t_{j=0}\left(\sum_{f \in a_i} \Tilde{y}_i^j (f) - \sum_{f^\prime \in a^\ast_i} \Tilde{y}_i^j (f^\prime) \right)$ is upper bounded by some small enough constant, then the induced policy of Algorithm \ref{alg} falls into the neighborhood set $U_\epsilon$. 

\begin{restatable}[]{lem}{nashUM}\label{lem:nash_UM}
    Let $\Tilde{z}_i^t(a_i) = \sum^t_{j=0}\left(\sum_{f \in a_i} \Tilde{y}_i^j (f) - \sum_{f^\prime \in a^\ast_i} \Tilde{y}_i^j (f^\prime) \right)$, and define 
    \[U_{M} = \left\{\omega^t \textup{ computed by Algorithm \ref{alg}} \mid \Tilde{z}_i^t(a_i) \leq - M \,, \forall a_i \neq a_i^\ast, \forall i \in [n] \right\} \,.\] 
    For sufficiently large $M$, $U_{M} \subseteq U_\epsilon$. Moreover, following the updates of Algorithm \ref{alg}, if 
    $\omega^t \in U_{M}$, then $\omega^{t+1} \in U_{M} $. 
\end{restatable}
Thus, if $\omega^0$ is in the neighborhood $U_{M} \subseteq U_\epsilon$, then the reward estimator $\Tilde{z}_i^t(a_i)$ can only decrease (by Lemma \ref{lem:neighborhood}), and hence the algorithm will give an updated policy $\omega^{t}$ that is also within the neighborhood set $U_\epsilon$. 

Also, note that $\omega_i^\ast$ is a strict Nash equilibrium, which implies that $\left| \omega_i^t - \omega_i^\ast\right|_1 = 2(1 - \omega_i^t(a_i^\ast))$. Hence, to establish the convergence rate, we need to lower bound
\begin{align*}
        \omega_i^t(a_i^\ast) 
        = \ & \frac{\prod_{f \in a_i^\ast} \Tilde{\omega}_i^t(f)}{\sum_{a^\prime \in \mathcal{A}} \prod_{f^\prime \in a^\prime} \Tilde{\omega}_i^t(f^\prime)} 
        =  \frac{\prod_{f \in a_i^\ast} \exp\left(\eta\sum^t_{j=0}\tilde{y}_i^j(f)\right)}{\sum_{a^\prime \in \mathcal{A}} \prod_{f^\prime \in a^\prime} \exp\left(\eta\sum^t_{j=0}\tilde{y}_i^j(f^\prime)\right)} \,.
\end{align*}

\paragraph{Technical challenge}
We remark that if we directly apply Lemma \ref{lem:nash_UM}, we can get
\begin{align*}
    \omega_i^t(a_i^\ast) 
        = \ & \frac{\prod_{f \in a_i^\ast} \Tilde{\omega}_i^t(f)}{\sum_{a^\prime \in \mathcal{A}} \prod_{f^\prime \in a^\prime} \Tilde{\omega}_i^t(f^\prime)} 
        \geq \  \frac{1}{1 + \sum_{a_i \in \mathcal{A}_i, a_i \neq a_i^\ast}\left(\prod_{f \in a_i} \Tilde{\omega}_i^t(f) - \prod_{f^\prime \in a_i^\ast} \Tilde{\omega}_i^t(f^\prime)\right)} \\
        \geq \ & 1 -  \sum_{a_i \in \mathcal{A}_i, a_i \neq a_i^\ast}\bigg(\prod_{f \in a_i} \Tilde{\omega}_i^t(f) - \prod_{f^\prime \in a_i^\ast} \Tilde{\omega}_i^t(f^\prime)\bigg) \,.
\end{align*}
Suppose one can upper bound $\sum_{a_i \in \mathcal{A}_i, a_i \neq a_i^\ast}\bigg(\prod_{f \in a_i} \Tilde{\omega}_i^t(f) - \prod_{f^\prime \in a_i^\ast} \Tilde{\omega}_i^t(f^\prime)\bigg) \leq \exp(-t)$, then this gives $1 - \sum_{a_i \in \mathcal{A}_i, a_i \neq a_i^\ast} \exp(-t)$, which yields a convergence rate of $(|\mathcal{A}_i|-1)\exp(-t)$ as $ \left\| \omega_i^t - \omega_i^\ast\right\|_1 = 2(1 - \omega_i^t(a_i^\ast)) $.
This approach implies that the convergence rate scales linearly with the number of actions (thus exponentially with the number of facilities), and is what we wanted to avoid in the analysis. 

To overcome this exponential dependency, we utilize the fact that any $k$-facility combination is an action, which means that we can order the facility from $f_1, \ldots, f_{F}$ in decreasing order of $\Tilde{\omega}_i^t(f)$ and $f_1, \ldots, f_k$ form the optimal pure Nash action $a_i^\ast$. 
\begin{proof}[Proof of Theorem~\ref{thm:nash}]
    Using the above-mentioned observation, we have,
    \begin{align}
        \nonumber \frac{\prod_{f \in a_i^\ast} \exp\left(\eta\sum^t_{j=0}\tilde{y}_i^j(f)\right)}{\sum_{a^\prime \in \mathcal{A}} \prod_{f^\prime \in a^\prime} \exp\left(\eta\sum^t_{j=0}\tilde{y}_i^t(f^\prime)\right)} 
        \geq \ & \prod_{m \in [k]} \frac{\exp\left(\eta\sum^t_{j=0}\tilde{y}_i^j(f_m)\right)}{\exp\left(\eta\sum^t_{j=0}\tilde{y}_i^j(f_m)\right) + \sum_{\ell > k} \exp\left(\eta\sum^t_{j=0}\tilde{y}_i^j(f_\ell)\right)} \\
        \nonumber = \ & \prod_{m \in [k]} \frac{1}{1 +  \sum_{\ell > k} \exp \left(\eta \sum^t_{j=0}\left(\tilde{y}_i^j(f_\ell) - \tilde{y}_i^j(f_m)\right)\right)} \\
        \label{Eq_888}\geq \ & \prod_{m \in [k]}  \Bigg(1 - \sum_{\ell > k} \exp \Bigg(\eta \sum^t_{j=0}\left(\tilde{y}_i^j(f_\ell) - \tilde{y}_i^j(f_m)\right)\Bigg)\Bigg) \,,
    \end{align}
    where the first inequality is by noting that any $\prod_{f^\prime \in a^\prime} \exp\left(\eta\sum^t_{j=0}\tilde{y}_i^t(f^\prime)\right)$ on the denominator of the left-hand side can be found on the product of the denominators of the right-hand side.  

    Consider an action $a_i$ formed from $a_i^\ast$ by replacing $f_m$ ($m\le k$) with $f_\ell$ ($\ell \ge k+1$), then
    \begin{align}
        \label{Eq_887}\Tilde{z}_i^t(a_i) = \sum^t_{j=0}\left(\sum_{f \in a_i} \Tilde{y}_i^j (f) - \sum_{f^\prime \in a^\ast_i} \Tilde{y}_i^j \left(f^\prime\right) \right)
        =  \sum^t_{j=0}\left(\Tilde{y}_i^j (f_{\ell})  -  \Tilde{y}_i^j (f_{m})\right)\,,
    \end{align}
    and 
    \begin{equation}\label{Eq_886}
        \Tilde{z}_i^{t}(a_i) 
        =  \Tilde{z}_i^{t-1}(a_i) + \eta \bigg(\sum_{f \in a_i} \Tilde{y}_i^t (f) - \sum_{f^\prime \in a^\ast_i} \Tilde{y}_i^t \left(f^\prime\right) \bigg)
        \leq  - M - \eta \epsilon t \,,
    \end{equation}
    %
    
    where the inequality is by the condition of starting with $\omega^0 \in U_{M}$ such that $M$ is large enough for $U_{M} \subseteq U_\epsilon$ and Lemma \ref{lem:nash_UM}. 

    Combining Equations (\ref{Eq_888}), (\ref{Eq_887}) and (\ref{Eq_886}), we have $\omega_i^t(a_i^\ast) \geq 1 - kF \exp(-M - \eta \epsilon t)$. Therefore, we obtain the result in Theorem \ref{thm:nash}, i.e., $\left\| \omega_i^t - \omega_i^\ast\right\|_1
        =  2(1 - \omega_i^t(a_i^\ast)) 
        \leq  2(kF \exp(-  M - \eta \epsilon t)) $.
\end{proof} 

Then we consider the case that each player observes only $R^f(a^t_i, a^t_{-i})$, instead of the expected rewards.  
\begin{restatable}[]{thm}{stocNASH}\label{thm:stoc_nash}
Consider the case where each player receives a stochastic reward under the full information setting. Assume the game permits a strict Nash equilibrium $\omega^\ast = (\omega_1^\ast, \cdots, \omega_n^\ast)$. Let $\Tilde{y}_i^{t}(f) = R^f(a^t_i, a^t_{-i})$ in Line 6 of Algorithm \ref{alg}, and set the learning rate to be time-dependent
such that $\sum^\infty_{t=0} \eta_t^2 \leq \frac{ \delta \cdot M^2}{8kn (F - 1)} \leq \sum^\infty_{t=0} \eta_t = \infty$.
Suppose $\Tilde{y}_i^0(f), \forall i \in [n]$ is initalized such that $\omega^0  \in U_{2M} \subseteq U_\epsilon$, then for any $i \in [n]$ and any $t$, we have \[\|\omega_i^t - \omega_i^\ast\|_1 \leq 2k F \exp \left(-M - \epsilon \sum^t_{j=0} \eta_j\right) \,,\] with probability at least $1 - \delta$, where $M \geq \left|\log\left(\frac{\epsilon}{2kF}\right)\right| $, and $\epsilon$ is a constant that is game-dependent only. 
\end{restatable}

We remark that in the case of stochastic rewards, the convergence rate of our algorithm cannot be arbitrarily large, as the learning rate $\eta$ cannot be taken to be arbitrarily large. If we take $\eta_t = \beta t^{-\alpha}$, with $\beta$ being a small positive constant and $\alpha \in (1/2, 1)$. Then our convergence rate is $\|\omega_i^t - \omega_i^\ast\|_1 \leq O \left(\exp \left(-\frac{\beta}{1 - \alpha}t^{1 - \alpha}\right)\right)$, which is close to exponentially fast convergence. When the reward function is smooth, we remark that this can imply each player only experience constant regret. 
\section{Conclusion}

This paper studied the congestion game under semi-bandit feedback and presented a modified version of the well-known exponential weights algorithm. The algorithm ensures sublinear regret for every player, with the regret being linearly dependent on the number of facilities. Additionally, the proposed algorithm can learn a policy that rapidly converges to the pure Nash policy, with the convergence rate also being linearly dependent on the number of facilities. To our best knowledge, these are the first results on congestion games for sublinear individual regret and geometric Nash convergence rate, without an exponential dependency on the number of facilities. 

There are several possible directions to further study the online congestion game. First, as our work only considers the semi-bandit feedback model for individual regret, the regret and convergence rate under the full-bandit feedback model remains unclear. For the Nash convergence result, our algorithm only enjoys theoretical guarantee in the full-information setting. It remains as future work to extend this result to semi-bandits and full-bandits feedback model. 
Moreover, it also remains in question whether the results of this work can be extended to the online Markov congestion game proposed by \citet{cuilearning22}.
\bibliography{ref}

\newpage
\appendix
\section*{Appendix}

\section{Proof of Theorem \ref{thm:regret}}
\regret*
\begin{proof}
 
Let $Q_i^t = \sum_{a_i \in \mathcal{A}_i} \prod_{f \in a_i} \Tilde{\omega}_i^t(f)$. Let $a_i^\ast$ denote the pure Nash equilibrium.

Notice that 
\begin{align*}
    \prod_{f \in a_i^\ast} \exp\left(\eta \sum^T_{j=1} \tilde{y}_i^j (f) \right)
    =  \prod_{f \in a_i^\ast} \Tilde{\omega}_i^T(f) 
    \leq \  Q_i^T 
    =  Q_i^0 \frac{Q_i^1}{Q_i^0} \ldots \frac{Q_i^T}{Q_i^{T-1}} = Q_i^0 \prod \frac{Q_i^t}{Q_i^{t-1}} \,.
\end{align*}

By the definition of $Q_i^t$ and the update rules of Algorithm \ref{alg}, we have
\begin{align*}
    \frac{Q_i^t}{Q_i^{t-1}} 
    = \ &  \frac{\sum_{a_i \in \mathcal{A}_i} \prod_{f \in a_i} \Tilde{\omega}_i^t(f)}{Q_i^{t-1}} \\
    = \ &  \sum_{a_i \in \mathcal{A}_i} \left(\frac{\exp \left(\eta \sum^{t-1}_{j=1} \sum_{f \in a_i} \tilde{y}_i^j(f)\right)}{Q_i^{t-1}}\right) \exp \left(\eta \sum_{f \in a_i} \tilde{y}_i^t(f)\right) \\
    = \ & \sum_{a_i \in \mathcal{A}_i} \left(\frac{\prod_{f \in a_i}\tilde{\omega}_i^{t-1}(f)}{Q_i^{t-1}}\right) \exp \left(\eta \sum_{f \in a_i} \tilde{y}_i^t(f)\right) \\
    = \ & \sum_{a_i \in \mathcal{A}_i} \left(\frac{\prod_{f \in a_i}\tilde{\omega}_i^{t-1}(f)}{Q_i^{t-1}}\right) \exp \left(\eta \sum_{f \in a_i} \tilde{y}_i^t(f)\right) \\
    = \ & \sum_{a_i \in \mathcal{A}_i} \omega_i^t(a_i) \exp \left(\eta \sum_{f \in a_i} \tilde{y}_i^t(f)\right) \\
    \leq \ &  \sum_{a_i \in \mathcal{A}_i} \omega_i^t(a_i) \left( 1+ \eta \sum_{f \in a_i} \tilde{y}_i^t(f) + \left(\eta \sum_{f \in a_i} \tilde{y}_i^t(f)\right)^2\right)\\
    =\ & 1 + \sum_{a_i \in \mathcal{A}_i} \omega_i^t(a_i) \left(\eta \sum_{f \in a_i} \tilde{y}_i^t(f) + \left(\eta \sum_{f \in a_i} \tilde{y}_i^t(f)\right)^2\right) \\
    \leq \ & \exp \left( \sum_{a_i \in \mathcal{A}_i} \omega_i^t(a_i) \left(\eta \sum_{f \in a_i} \tilde{y}_i^t(f) + \left(\eta \sum_{f \in a_i} \tilde{y}_i^t(f)\right)^2\right) \right) \,,
\end{align*}
where the first inequality is by $\eta \leq \frac{1}{k}$, $\exp(x) \leq 1 + x + x^2$ for $x \leq 1$ and the second inequality is by $1 + x \leq \exp(x)$ for $x \in \mathbb{R}$.

By initialization, we also have
\begin{align*}
    Q_i^0=\sum_{a_i \in \mathcal{A}_i} \exp \left(\eta \sum_{f \in a_i} \tilde{y}_i^0(f)\right) \leq A_i \,.
\end{align*}

Thus we have
\begin{align*}
    \prod_{f \in a_i^\ast} \exp\left(\eta \sum^T_{j=1} \tilde{y}_i^j (f)\right) 
    = \ & \exp\left(\eta \sum^T_{j=1} \sum_{f \in a_i^\ast}\tilde{y}_i^j (f)\right) \\
    \leq \ & A_i \prod^T_{t=1} \exp \left( \sum_{a_i \in \mathcal{A}_i} \omega_i^t(a_i) \left(\eta \sum_{f \in a_i} \tilde{y}_i^t(f) + \left(\eta \sum_{f \in a_i} \tilde{y}_i^t(f)\right)^2\right) \right) \\
    = \ & A_i \exp \left( \sum^T_{t=1} \sum_{a_i \in \mathcal{A}_i} \omega_i^t(a_i) \left(\eta \sum_{f \in a_i} \tilde{y}_i^t(f) + \eta^2 \left(\sum_{f \in a_i} \tilde{y}_i^t(f)\right)^2\right) \right) \,.
\end{align*}

Taking logarithms on both sides, and noticing that $\log \left(A_i\right) \leq k F $ we have
\begin{align*}
    \eta \sum^T_{j=1} \sum_{f \in a_i^\ast}\tilde{y}_i^j (f)
    \leq \ & k F +  \sum^T_{t=1} \sum_{a_i \in \mathcal{A}_i} \omega_i^t(a_i) \left(\eta \sum_{f \in a_i} \tilde{y}_i^t(f) + \eta^2 \left(\sum_{f \in a_i} \tilde{y}_i^t(f)\right)^2\right) \,.
\end{align*}

Rearranging the terms, we have
\begin{align*}
    \sum^T_{t=1} \sum_{f \in a_i^\ast}\tilde{y}_i^t (f) - \sum_{a_i \in \mathcal{A}_i} \omega_i^t(a)  \sum_{f \in a_i}\tilde{y}_i^t (f) 
    \leq  \frac{k F }{\eta} + \eta \sum^T_{t=1} \sum_{a_i \in \mathcal{A}_i} \omega_i^t(a_i) \left(\sum_{f \in a_i} \tilde{y}_i^t(f)\right)^2 \,.
\end{align*}

It thus amounts to bound the second term. By the definition of $\tilde{y}_i^t$, we have
\begin{align*}
   & \sum_{a_i \in \mathcal{A}_i} \omega_i^t(a_i) \left(\sum_{f \in a_i} \tilde{y}_i^t(f)\right)^2 \\
    = \ &  \sum_{a_i \in \mathcal{A}_i} \omega_i^t(a_i) \left(\sum_{f \in a_i} 1 - \frac{\mathbb{I} \{f \in a_i^t\}(1 - R^f(a^t_i, a^t_{-i}))}{q_i^t(f)}\right)^2 \\
    \leq \ & k \sum_{a_i \in \mathcal{A}_i} \omega_i^t(a_i) \sum_{f \in a_i} \left(1 - \frac{\mathbb{I} \{f \in a_i^t\}(1 - R^f(a^t_i, a^t_{-i}))}{q_i^t(f)}\right)^2 \\
= \ & k \sum_{a_i \in \mathcal{A}_i}  \omega_i^t(a_i) \sum_{f \in a_i}\left(1-\frac{2 \mathbb{I}\left\{f \in a_i^t\right\}\left(1-R_i^f\left(a_i^t, a_i^t\right)\right)}{q_i^t(f)}+\left(\frac{\mathbb{I}\left\{f \in a_i^t\right\}\left(1-R_i^f\left(a_i^t, a_{-i}^t\right)\right)}{q_i^t(f)}\right)^2\right) \\
= \ & k+k\sum_{a_i \in \mathcal{A}_i}  \omega_i^t(a_i) \sum_{f \in a_i}\left(\left(\frac{\mathbb{I}\left\{f \in a_i^t\right\}\left(1-R_i^f\left(a_i^t, a_{-i}^t\right)\right)}{q_i^t(f)}\right)^2-\frac{2 \mathbb{I}\left\{f \in a_i^t\right\}\left(1-R_i^f\left(a_i^t, a_{-i}^t\right)\right)}{q_i^t(f)}\right) \\
= \ & k+k\sum_{f \in \mathcal{F}}\left(\left(\frac{\mathbb{I}\left\{f \in a_i^t\right\}\left(1-R_i^f\left(a_i^t, a_{-i}^t\right)\right)}{q_i^t(f)}\right)^2-\frac{2 \mathbb{I}\left\{f \in a_i^t\right\}\left(1-R_i^f\left(a_i^t, a_{-i}^t\right)\right)}{q_i^t(f)}\right) \sum_{a_i \in \mathcal{A}, f \in a_i} \omega_i^t(a_i)\\
= \ & k+k\sum_{f \in \mathcal{F}}\left(\left(\frac{\mathbb{I}\left\{f \in a_i^t\right\}\left(1-R_i^f\left(a_i^t, a_{-i}^t\right)\right)}{q_i^t(f)}\right)^2-\frac{2 \mathbb{I}\left\{f \in a_i^t\right\}\left(1-R_i^f\left(a_i^t, a_{-i}^t\right)\right)}{q_i^t(f)}\right) q_i^t(f)\\
= \ & k+k\sum_{f \in \mathcal{F}}\left(\left(\frac{\mathbb{I}\left\{f \in a_i^t\right\}\left(1-R_i^f\left(a_i^t, a_{-i}^t\right)\right)}{q_i^t(f)}\right)^2 q_i^t(f) -2 \mathbb{I}\left\{f \in a_i^t\right\}\left(1-R_i^f\left(a_i^t, a_{-i}^t\right)\right)\right) \\
\leq \ & k+k\sum_{f \in \mathcal{F}}\left(\frac{\mathbb{I}\left\{f \in a_i^t\right\}\left(1-R_i^f\left(a_i^t, a_{-i}^t\right)\right)}{q_i^t(f)}\right)^2 q_i^t(f)\,.
\end{align*}

Let $\mathbb{E}_t[\cdot]$ denote conditional expectation over all history up to time $t$. Notice that our estimator $\tilde{y}_i^t(f)$ is unbiased. 
Taking expectations on both sides yield
\begin{align*}
    \mathbb{E}\left[\sum_{a_i \in \mathcal{A}_i} \omega_i^t(a_i) \left(\sum_{f \in a_i} \tilde{y}_i^t(f)\right)^2 \right] 
    \leq \ & 
    k + k\mathbb{E}\left[\sum_{f \in \mathcal{F}}\left(\frac{\mathbb{I}\left\{f \in a_i^t\right\}\left(1-R_i^f\left(a_i^t, a_{-i}^t\right)\right)}{q_i^t(f)}\right)^2 q_i^t(f) \right]\\
    \leq \ & k + k\mathbb{E}\left[\mathbb{E}_{t-1} \left[\sum_{f \in \mathcal{F}}\frac{\mathbb{I}\left\{f \in a_i^t\right\}}{q_i^t(f)} \right] \right] \\
    \leq \ & k + kF \,.
\end{align*}

Combining the terms, and by the unbiasedness of our estimator $\Tilde{y}_i^t$, we have
\begin{align*}
    \text{Regret}_i(T) 
    = \ & \mathbb{E} \left[\sum^T_{t=1} \sum_{f \in a_i^\ast}\tilde{y}_i^t (f) - \sum_{a_i \in \mathcal{A}_i} \omega_i^t(a)  \sum_{f \in a_i}\tilde{y}_i^t (f) \right] \\
    \leq \ & \frac{k F }{\eta} + \eta T \left(k + kF\right) \\
    = \ & O \left(\frac{k F }{\eta} + \eta T \left(kF\right)\right)\,.
\end{align*}

Taking $\eta = \frac{1}{\sqrt{T}}$ gives
\begin{align*}
    \text{Regret}_i(T)  = O \left(k F \sqrt{T} \right) \,.
\end{align*}
\end{proof}
\newpage
\section{Proof of Corollary \ref{cor:poa}}
\poa*
\begin{proof}
    The proof largely follows from Proposition 2 of \citep{syrgkanis2015fast}.

    Denote $a_i^\ast = \argmax_{\omega_i \in \gA_i} \sum^T_{t=1} r_i(\omega_i, \omega_{-i}^t) - r_i(\omega_i^t, \omega_{-i}^t)$.

    For any $i \in [n]$ and policy $\omega = (\omega_i, \omega_{-i})$, by Definition \ref{def:individual_regret}, we have
    \begin{align*}
        \sum^T_{t=1} r_i(\omega_i, \omega_{-i}) \geq \sum^T_{t=1} r_i(a_i^\ast, \omega_{-i}^t) - \mathrm{Regret}_i(T) \,.
    \end{align*}

    Summing over $i \in [n]$, we have
    \begin{align*}
        \sum^n_{i=1} \sum^T_{t=1} r_i(\omega_i, \omega_{-i}) \geq \ & \sum^n_{i=1}\sum^T_{t=1} r_i(a_i^\ast, \omega_{-i}^t) -  \sum^n_{i=1}\mathrm{Regret}_i(T) \\
        = \ & \sum^n_{i=1}\sum^T_{t=1} \mathbb{E}_{a_{-i}^t \sim \omega^t} \left[r_i(a_i^\ast, a_{-i}^t) \right] - \sum^n_{i=1}\mathrm{Regret}_i(T) \\
        \geq \ & \lambda \cdot \mathrm{OPT} - \mu \mathbb{E}_{a \sim \omega^t}\left[W(a)\right] - \sum^n_{i=1}\mathrm{Regret}_i(T) \,,
    \end{align*}
    where the second inequality is by Definition \ref{def:smooth}.

    Rearranging the terms and substituting results from Theorem \ref{thm:regret}, we have
    \begin{align*}
        \frac{1}{T} \sum^T_{t=1} W(\omega^t) \geq \frac{\lambda}{1 + \mu} \mathrm{OPT} - \frac{1}{T(1+\mu)} \sum^n_{i=1} O \left(k \log \left(A_i\right) \sqrt{T} + F\sqrt{T}\right) \,.
    \end{align*}
    
\end{proof}

\newpage
\section{Proof of Lemma \ref{lem:nash_UM}}
\nashUM*

\begin{proof}
    \begin{align*}
        \omega_i^t(a_i^\ast) 
        = \ & \frac{\prod_{f \in a_i^\ast} \Tilde{\omega}_i^t(f)}{\sum_{a^\prime \in \mathcal{A}} \prod_{f^\prime \in a^\prime} \Tilde{\omega}_i^t\left(f^\prime\right)} 
        =  \frac{\prod_{f \in a_i^\ast} \exp\left(\eta\sum^t_{j=0}\tilde{y}_i^j(f)\right)}{\sum_{a^\prime \in \mathcal{A}} \prod_{f^\prime \in a^\prime} \exp\left(\eta\sum^t_{j=0}\tilde{y}_i^j\left(f^\prime\right)\right)} \,.
    \end{align*}
    Without loss of generality, we order the facility from $1, \ldots, |F|$, such that $f_1, \ldots, f_k$ form the optimal pure Nash action $a_i^\ast$.
    We have
    \begin{align*}
        \frac{\prod_{f \in a_i^\ast} \exp\left(\eta\sum^t_{j=0}\tilde{y}_i^j(f)\right)}{\sum_{a^\prime \in \mathcal{A}} \prod_{f^\prime \in a^\prime} \exp\left(\eta\sum^t_{j=0}\tilde{y}_i^j\left(f^\prime\right)\right)}  
        \geq \ & \prod_{m \in [k]} \frac{\exp\left(\eta\sum^t_{j=0}\tilde{y}_i^j\left(f_m\right)\right)}{\exp\left(\eta\sum^t_{j=0}\tilde{y}_i^j\left(f_m\right)\right) + \sum_{\ell > k} \exp\left(\eta\sum^t_{j=0}\tilde{y}_i^j\left(f_\ell\right)\right)} \,,
    \end{align*}
    where the inequality is by noting that any $\prod_{f^\prime \in a^\prime} \exp\left(\eta\sum^t_{j=0}\tilde{y}_i^t(f^\prime)\right)$ on the denominator of the left-hand side can be found on the product of the denominators of the right-hand side. 
    We then have
    \begin{align*}
        &\prod_{m \in [k]} \frac{\exp\left(\eta\sum^t_{j=0}\tilde{y}_i^j\left(f_m\right)\right)}{\exp\left(\eta\sum^t_{j=0}\tilde{y}_i^j\left(f_m\right)\right) + \sum_{\ell > k} \exp\left(\eta\sum^t_{j=0}\tilde{y}_i^j\left(f_\ell\right)\right)} \\
        = \ & \prod_{m \in [k]} \frac{1}{1 +  \sum_{\ell > k} \exp \left(\eta \sum^t_{j=0}\left(\tilde{y}_i^j\left(f_\ell\right) - \tilde{y}_i^j\left(f_m\right)\right)\right)} \\
        \geq \ & \prod_{m \in [k]}  \left(1 - \sum_{\ell > k} \exp \left(\eta \sum^t_{j=0}\left(\tilde{y}_i^j\left(f_\ell\right) - \tilde{y}_i^j\left(f_m\right)\right)\right)\right) \\
        \geq \ & \prod_{m \in [k]}  \left(1 - \sum_{\ell > k} \exp \left( - \eta M\right)\right) \\
        \geq \ & 1 - k F \exp \left( - \eta M\right)\,,
    \end{align*}
    where the second inequality is by considering an action $a_i$ formed from $a_i^\ast$ by replacing $f_m$ with $f_\ell$, then
    \begin{align*}
        \Tilde{z}_i^t(a_i) = \sum^t_{j=0}\left(\sum_{f \in a_i} \Tilde{y}_i^j (f) - \sum_{f^\prime \in a^\ast_i} \Tilde{y}_i^j \left(f^\prime\right) \right)
        = \ & \sum^t_{j=0}\left(\sum_{f \in a_i, f \neq f_m}\Tilde{y}_i^j (f) + \Tilde{y}_i^j (f_{\ell})  - \sum_{f \in a_i^\ast} \Tilde{y}_i^j (f) \right)\\
        = \ & \sum^t_{j=0}\left(\Tilde{y}_i^j (f_{\ell})  -  \Tilde{y}_i^j (f_{m})\right)
        \leq \ - M \,.
    \end{align*}
    The last inequality is by the condition of $z_i^t(a_i) \leq -M$.
    
    Therefore, 
    \begin{align*}
        \left\| \omega_i^t - \omega_i^\ast\right\|_1
        = \ & 2(1 - \omega_i^t(a_i^\ast)) \\
        \leq \ &  2(kF \exp(- M)) \,.
    \end{align*}
    To have $U_{M} \subseteq U_\epsilon$, for any $\epsilon$, we let  
    \begin{align*}
        2(k F \exp(- M)) \leq \ & \epsilon \,, \\
        M \geq \ & \left|\log\left(\frac{\epsilon}{2kF}\right)\right| \,.
    \end{align*}
    For the second claim, notice that
    \begin{align*}
        \Tilde{z}_i^{t+1}(a_i) 
        = \ &  \Tilde{z}_i^{t}(a_i) + \eta \left(\sum_{f \in a_i}\tilde{y}_i^t(f) - \sum_{f^\ast \in a_i^\ast}\tilde{y}_i^t(f^\ast )\right) \\
        \leq \ & \Tilde{z}_i^{t}(a_i) - \eta \epsilon \\
        \leq \ & -M - \eta \epsilon \,,
    \end{align*}
    where the first inequality is due to the definition of $U$.
    Since $\omega_i^{t+1}$ is induced by $\Tilde{z}_i^{t+1}(a_i) $, we have $\omega_i^{t+1} \in U_{M}$.
\end{proof}

\newpage
\section{Proof of Theorem \ref{thm:stoc_nash}}
\stocNASH*
\begin{proof}
Define $r_i^f(a_i, \omega_{-i}) = \mathbb{E}_{a_{-i} \sim  \omega_{-i}}[r_i^f(a_i, a_{-i})]$. 
Let  $x_i^t(f) = \tilde{y}_i^t(f) -  r_i^f(a_i^t, \omega_{-i}^t)$,
then we have
\begin{align*}
    \tilde{z}_i^{t+1}(a_i) 
    = \ & \tilde{z}_i^{t}(a_i)  + \eta_t \left(\sum_{f \in a_i} \tilde{y}_i^t(f) - \sum_{f \in a^\ast_i} \tilde{y}_i^t(f) \right)\\
    = \ & \tilde{z}_i^{t}(a_i) + \eta_t \left(\sum_{f \in a_i} r_i^f(a_i^t, \omega_{-i}^t) - \sum_{f \in a^\ast_i} r_i^f(a_i^t, \omega_{-i}^t)\right) \\
    & + \eta_t \left(\left(\sum_{f \in a_i} \tilde{y}_i^t(f) - \sum_{f \in a^\ast_i} \tilde{y}_i^t(f) \right) - \left(\sum_{f \in a_i} r_i^f(a_i^t, \omega_{-i}^t) - \sum_{f \in a^\ast_i} r_i^f(a_i^t, \omega_{-i}^t)\right) \right)\,.
\end{align*}

    Define 
    \[U_{2M} = \left\{\omega^t \text{ computed by Algorithm \ref{alg} } \mid \Tilde{z}_i^t(a_i) \leq - 2M \,, \forall a_i \neq a_i^\ast, \forall i \in [n] \right\} \,.\] 

    Similar to Lemma \ref{lem:nash_UM}, suppose that we initialize the algorithm such that $\Tilde{z}_i^0(a_i) \leq - 2M$, $\forall i \in [n]$ and with sufficiently large $M$, $U_{2M} \subseteq U_\epsilon$. 
    Then, we have
    \begin{align} \label{eq:stoc_z}
        \tilde{z}_i^{t+1}(a_i) 
        = \ & \tilde{z}_i^{t}(a_i) + \sum^t_{j=0}\eta_j\left(\sum_{f \in a_i} r_i^f(a_i^j, \omega_{-i}^j) - \sum_{f \in a^\ast_i} r_i^f(a_i^j, \omega_{-i}^j)\right) + 2k \max_f\sum^t_{j=0} \eta_j  x_i^t(f) \nonumber\\
        \leq \ & - 2M_i +\sum^t_{j=0} \eta_j\left(\sum_{f \in a_i} r_i^f(a_i^j, \omega_{-i}^j) - \sum_{f \in a^\ast_i} r_i^f(a_i^j, \omega_{-i}^j)\right) +  2k \max_f\sum^t_{j=0} \eta_j  x_i^t(f) \,.
    \end{align}

    Let $X_i^t(f) =\sum^t_{j = 1} \eta_j x_i^t(f) $, then $X_i^t(f)$ is a Martingale. By Doob's maximal inequality, we have
    \begin{align*}
        \mathbb{P} \left[\max_{j \in [1, t]}\left|X_i^j(f) \right| \geq M/2k\right] \leq \frac{4k^2\mathbb{E}\left[X_i^t(f)^2\right]}{\left(M\right)^2} \leq \frac{4k^2 \cdot j\sum^t_{j=0} \eta_j^2 }{\left(M\right)^2} \,.
    \end{align*}

    Notice that whenever $\mathbb{I}\left\{\max_{j \in [1, t]}\left|X_i^j(f) \right| \geq M/2k\right\} = 1$, $\mathbb{I}\left\{\max_{j \in [1, t+1]}\left|X_i^j(f) \right| \geq M/2k\right\} = 1$. 
    Thus, by choosing $\sum^\infty_{j=0} \eta_j^2 \leq \frac{ \delta \cdot M_i^2}{8n k^2(F - 1) }$, $\delta > 0$, we have
    \begin{align*}
        \mathbb{P}\left[\left|X_i^j(f) \right| \geq M/2k \,, \forall t\right] \leq \frac{\delta}{n (F - 1) } \,.
    \end{align*}

    Then we obtain
    \begin{align*}
        \mathbb{P}\left[\max_{i \in [n], f \in \mathcal{F}}\sup_t\left|X_i^j(f) \right| \geq M/2k \right] \leq \delta \,.
    \end{align*}

    With this and Eq. (\ref{eq:stoc_z}), we can use of Lemma \ref{lem:nash_UM} to show by induction that if $\omega^t \in U_{M}$, so is $\omega^{t+1} $, with probability at least $1 - \delta$. 
    Then, for any $i \in [n]$, we have
    \begin{align*}
        \mathbb{P} \left[\tilde{z}_i^{t+1} (a_i) \leq -M + \epsilon \sum^t_{j=0} \eta_j \right] \leq \delta \,.
    \end{align*}

    Similar to the argument of Theorem \ref{thm:nash}, we have
    \begin{align*}
        \omega_i^t(a_i^\ast)
        \geq \ & 1 - k F \exp \left(-M + \epsilon \sum^t_{j=0} \eta_j\right) \,,
    \end{align*}
    with probability at least $1 - \delta$.

    Hence, with a probability of at least $1 - \delta$, we have
    \begin{align*}
        \|\omega_i^t - \omega_i^\ast\|_1 \leq 2 k F \exp \left(-M + \epsilon \sum^t_{j=0} \eta_j\right)\,.
    \end{align*}
\end{proof}

\end{document}